\documentclass[runningheads]{llncs}
\usepackage{graphicx}
\usepackage[utf8]{inputenc}
\usepackage{natbib}
\usepackage{verbatim}
\usepackage{setspace}
\usepackage{amssymb}
\usepackage{amsmath}
\usepackage{mathtools}
\usepackage{tensor}
\usepackage{pgf}
\usepackage{tikz}
\usepackage{graphicx}
\graphicspath{ {./images/} }
\usepackage{breqn}
\usepackage[margin=1in]{geometry}
\usepackage{hyperref}
\usepackage{nameref}
\usepackage[shortlabels]{enumitem}
\usepackage[ruled, noend]{algorithm2e}
\SetKwComment{Comment}{$\triangleright$\ }{}

\usepackage{xcolor}

\usetikzlibrary{decorations}
\usetikzlibrary[decorations]
\usepgflibrary{decorations.pathmorphing}
\usepgflibrary[decorations.pathmorphing]
\usetikzlibrary{decorations.pathmorphing}
\usetikzlibrary[decorations.pathmorphing]
\usepackage{tikzit}

\tikzstyle{standard}=[fill=white, draw=black, shape=circle, scale=1]
\tikzstyle{label}=[fill=white, draw=white, scale=1]

\tikzstyle{weighted directed}=[->]
\tikzstyle{squiggly}=[->, decoration={{snake}}, tikzit draw={rgb,255: red,191; green,255; blue,0}, decorate]

\usepackage[font={small}]{caption}
\usepackage{subcaption}
\usepackage{float}

\newcounter{fig}
\usepackage{todonotes}

\hypersetup{
     colorlinks=true,      		
     linkcolor=blue,        
     citecolor=magenta,     
     filecolor=green,      		
     urlcolor=orange,           	
}

\DeclareMathOperator*{\argmin}{argmin}
\DeclareMathOperator*{\E}{\mathbb{E}}
\DeclareMathOperator*{\len}{len}
\newcommand{\eps}{\varepsilon}

\newcommand{\sP}{\mathcal{P}}
\begin{document}
\title{Competition Alleviates Present Bias in Task Completion}
%
%
\author{Aditya Saraf 
\thanks{Research supported in part by NSF grant CCF-1813135.} \and 
Anna R. Karlin\thanks{Research supported by NSF grant CCF-1813135.} \and
Jamie Morgenstern
}
%
%
\institute{University of Washington, Seattle, WA, USA\\
\email{\{sarafa,karlin,jamiemmt\}@cs.washington.edu}}
\maketitle              
\begin{abstract}
We build upon recent work \citep{kleinberg2014naive,kleinberg2016sophisticated,kleinberg2017multiple} that considers \emph{present biased} agents, who place more weight on costs they must incur now than costs they will incur in the future. They consider a graph theoretic model where agents must complete a task and show that present biased agents can take exponentially more expensive paths than optimal. We propose a theoretical model that adds \emph{competition} into the mix -- two agents compete to finish a task first. We show that, in a wide range of settings, a small amount of competition can alleviate the harms of present bias. This can help explain why biased agents may not perform so poorly in naturally competitive settings, and can guide task designers on how to protect present biased agents from harm. Our work thus paints a more positive picture than much of the existing literature on present bias.

\end{abstract}

\section{Introduction}

One of the most influential lines of recent economic research has been \textit{behavioral} game theory \citep{ariely2008predictably,kahneman1979prospect}. 
The majority of economics research makes several idealized assumptions about the behavior of rational agents to prove mathematical results. Behavioral game theory questions these assumptions and proposes models of agent behavior that more closely align with human behavior. 
Through experimental research~\citep{dellavigna2006paying,dellavigna2009psychology}, behavioral economists have observed and codified several common types of cognitive biases, from loss aversion \citep{kahneman1979prospect} (the tendency to prefer avoiding loss to acquiring equivalent gains) to the sunk cost fallacy \citep{arkes1985psychology} (the tendency to factor in previous costs when determining the best future course of action) to present bias \citep{frederick2002time} (the current topic). 
One primary goal of theorems in game theory  is to offer predictive power. This perspective is especially important in the many computer science applications of these results, from modern ad auctions to cryptocurrency protocols. If these theorems are to predict human behavior, the mathematical models ought to include observed human biases. Thus, rather than viewing behavioral game theory as conflicting with the standard mathematical approach, the experimental results of behavioral game theory can inform more sophisticated mathematical models. This paper takes a step towards this goal, building on seminal work of~\citet{kleinberg2014naive} and~\citet{kleinberg2016sophisticated,kleinberg2017multiple}, who formulated a mathematical model for planning problems where agents are present biased.

Present bias refers to overweighting immediate costs relative to future costs. This is a ubiquitous bias in human behavior that explains diverse phenomena. The most natural example is procrastination, the familiar desire to delay difficult work, even when this predictably leads to negative consequences later. Present bias can also model the tendency of firms to prefer immediate gains to long-term gains and the tendency of politicians to prefer immediate results to long-term plans. One simple model of present bias\citep{kleinberg2014naive,kleinberg2016sophisticated,kleinberg2017multiple} is to multiply costs in the current time period by present bias parameter $b$ when making plans. This model is a special case of hyperbolic discounting, where costs are discounted in proportion to how much later one would experience them. But even this special case suffices to induces \emph{time-inconsistency}, resulting in a rich set of strategies consistent with human behavior.

Examples of time inconsistent behavior extend beyond procrastination. For example, one might undertake a project, and abandon it partway through, despite the underlying cost structure remaining unchanged. One might fail to complete a course with no deadlines, but pass the same course with weekly deadlines. Many people pay for a gym membership but never use it. \citet{kleinberg2014naive} presented the key insight that this diverse range of phenomena can all be expressed in a single graph-theoretic framework, which we describe below.

Fix a directed, acyclic graph $G$, with designated source $s$ and sink $t$. Refer to $G$ as a task graph, where $s$ is the start of the task and $t$ the end. A path through this graph corresponds to a plan to complete the task; each edge represents one step of the plan. Each edge has a weight corresponding to the cost of that step. 

The goal of an agent is to complete the task while incurring the least cost (i.e., to take the cheapest path from $s$ to $t$). An optimal agent will simply follow such a cheapest path. A \textit{naive} present biased agent with bias parameter $b$ behaves as follows. At $s$, they compute their perceived cost for taking each path to $t$ by summing the weights along this path with the first edge scaled up by $b > 1$. They choose the path with the lowest perceived cost and take \textit{one} step along this path, say to $v$, and recompute their perceived cost along each the path from $v$ to $t$. Notice that such an agent may choose a path at $s$, take one edge along that path, and then deviate away from it. This occurs because the agent believes that, after the current choice of edge, they will pick the path from $v$ to $t$ with lowest true cost. But once they arrive at $v$, their perceived cost of a path differs from the true cost, and they pick a path with lowest perceived cost. This is why the agents are considered naive: they incorrectly assume their future self will behave optimally, and thus exhibit time-inconsistent behavior. See \autoref{fig:branching} for an example.

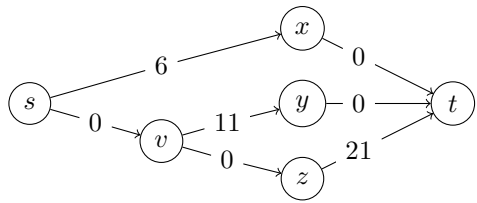
\begin{figure}[b]
    \centering
    \begin{tikzpicture}
	\begin{pgfonlayer}{nodelayer}
		\node [style=standard] (0) at (0, 0) {$s$};
		\node [style=standard] (2) at (7.25, 2) {$x$};
		\node [style=standard] (5) at (3.5, -1) {$v$};
		\node [style=label] (6) at (5.25, -1.5) {$0$};
		\node [style=standard] (7) at (11.25, 0) {$t$};
		\node [style=standard] (8) at (7.25, -2) {$z$};
		\node [style=label] (9) at (8.75, -1.25) {$21$};
		\node [style=standard] (13) at (7.25, 0) {$y$};
		\node [style=label] (14) at (3.5, 1) {$6$};
		\node [style=label] (17) at (1.75, -0.5) {$0$};
		\node [style=label] (18) at (8.75, 0) {$0$};
		\node [style=label] (19) at (8.75, 1.25) {$0$};
		\node [style=label] (20) at (5.25, -0.5) {$11$};
	\end{pgfonlayer}
	\begin{pgfonlayer}{edgelayer}
		\draw (5) to (6);
		\draw [style=weighted directed] (6) to (8);
		\draw (8) to (9);
		\draw [style=weighted directed] (9) to (7);
		\draw (0) to (14);
		\draw (2) to (19);
		\draw [style=weighted directed] (19) to (7);
		\draw (13) to (18);
		\draw [style=weighted directed] (18) to (7);
		\draw (0) to (17);
		\draw [style=weighted directed] (17) to (5);
		\draw [style=weighted directed] (14) to (2);
		\draw (5) to (20);
		\draw [style=weighted directed] (20) to (13);
	\end{pgfonlayer}
\end{tikzpicture}
    \caption{The optimal path is $(s, x, t)$ with total cost $6$. However, an agent with bias $b = 2$ will take path $(s, v, z, t)$, with cost $21$. Importantly, when the agent is deciding which vertex to move to from $s$, they evaluate $x$ as having total cost $12$, while $v$ has total cost $11$. This is because they assume they will behave optimally at $v$ by taking path $(v, y, t)$. However, they apply the same bias at $v$ and deviate to the worst possible path.}
    \label{fig:branching}
\end{figure}
The power of this graph theoretic model is that it allows us to answer questions over a range of planning problems, and to formally investigate which tasks represent the ``worst-case'' cost of procrastination. This is useful both to understand how present-biased behavior differs from optimal behavior and to design tasks to accommodate present bias. We now briefly summarize the existing literature, to motivate our introduction of competition to the model.

\subsection{Prior Work}
The most striking result is that there are graphs where the \textit{cost ratio} (the ratio of the optimal agent's cost to the biased agent's cost) is exponential in the size of the graph. In addition, all graphs with exponential cost ratio have a shared structure -- they all have a large $n$-fan as a graph minor (and graphs without exponential cost ratio do not) \citep{kleinberg2014naive,tang2017computational}. So this structure encodes the worst-case behavior for present bias in the standard model (and we later show how competition is especially effective in this graph). An $n$-fan is pictured in \autoref{fig:n-fan}. 

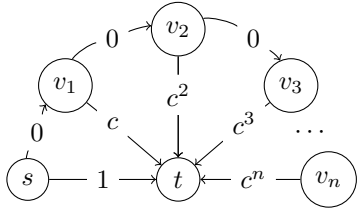
\begin{figure}
    \centering
    \begin{tikzpicture}
	\begin{pgfonlayer}{nodelayer}
		\node [style=standard] (0) at (1, 0) {$s$};
		\node [style=standard] (4) at (2, 2.5) {$v_1$};
		\node [style=label] (5) at (1.25, 1.25) {$0$};
		\node [style=label] (6) at (3.25, 3.75) {$0$};
		\node [style=label] (8) at (3.25, 1.5) {$c$};
		\node [style=label] (9) at (3, 0) {$1$};
		\node [style=standard] (10) at (9, 0) {$v_n$};
		\node [style=standard] (11) at (5, 0) {$t$};
		\node [style=standard] (12) at (5, 4) {$v_2$};
		\node [style=standard] (13) at (8, 2.5) {$v_3$};
		\node [style=label] (14) at (8.5, 1.25) {$\cdots$};
		\node [style=label] (15) at (7, 3.75) {$0$};
		\node [style=label] (16) at (5, 2.25) {$c^2$};
		\node [style=label] (17) at (6.75, 1.5) {$c^3$};
		\node [style=label] (18) at (7, 0) {$c^n$};
	\end{pgfonlayer}
	\begin{pgfonlayer}{edgelayer}
		\draw [bend left=15, looseness=0.50] (0) to (5);
		\draw [bend left, looseness=0.50] (4) to (6);
		\draw (4) to (8);
		\draw (0) to (9);
		\draw [style=weighted directed, bend left=15, looseness=0.75] (5) to (4);
		\draw (12) to (16);
		\draw (13) to (17);
		\draw (10) to (18);
		\draw [bend right, looseness=0.50] (15) to (12);
		\draw [style=weighted directed] (18) to (11);
		\draw [style=weighted directed] (17) to (11);
		\draw [style=weighted directed] (16) to (11);
		\draw (12) to (16);
		\draw [style=weighted directed, bend left=15, looseness=0.50] (6) to (12);
		\draw [style=weighted directed] (9) to (11);
		\draw [style=weighted directed] (8) to (11);
		\draw [style=weighted directed] (16) to (11);
		\draw [style=weighted directed, bend left=15, looseness=0.75] (15) to (13);
	\end{pgfonlayer}
\end{tikzpicture}
    \caption{A naive agent with bias $b > c > 1$ will continually choose to delay finishing the task.}
    \label{fig:n-fan}
\end{figure}

The exponential cost ratio demonstrates the severe harm caused by present bias. How, then, can designers of a task limit the negative effects of present bias? \citet{kleinberg2014naive} propose a model where a reward is given after finishing the task, and where the agent will abandon the task if at any point, they perceive the remaining cost to be higher than the reward. Unlike an optimal agent, a biased agent may abandon a task partway through. \autoref{fig:gym} uses the gym membership example from earlier to show this. As a result, they give the task designer the power to arbitrarily delete vertices and edges, which can model deadlines, as \autoref{fig:class} shows. 
They then investigate the structure of \textit{minimally motivating subgraphs} -- the smallest subgraph where the agent completes the task, for some fixed reward. Follow-up work of \citet{tang2017computational}  shows that finding \textit{any} motivating subgraph is NP-hard. Instead of deleting edges, \citet{albers2019motivating} consider the problem of spreading a fixed reward onto arbitrary vertices to motivate an agent, and find that this too is NP-hard (with a constrained budget). For other recent work involving present bias, see \citep{gravin2016procrastination,albers2017value,yan2019time,oren2019principal,ma2019penalty}.

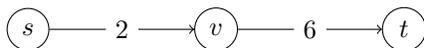
\begin{figure}[b]
    \centering
    \begin{tikzpicture}
	\begin{pgfonlayer}{nodelayer}
		\node [style=standard] (0) at (0.5, 0) {$s$};
		\node [style=label] (1) at (3, 0) {$2$};
		\node [style=standard] (2) at (10.5, 0) {$t$};
		\node [style=standard] (3) at (5.5, 0) {$v$};
		\node [style=label] (4) at (8, 0) {$6$};
	\end{pgfonlayer}
	\begin{pgfonlayer}{edgelayer}
		\draw (0) to (1);
		\draw [style=weighted directed, in=180, out=0] (1) to (3);
		\draw (3) to (4);
		\draw [style=weighted directed] (4) to (2);
	\end{pgfonlayer}
\end{tikzpicture}
    \caption{Let $(s, v)$ represent buying a gym membership and $(v, t)$ represent working out regularly for a month \citep{roughgarden2016cs269i}. At $t$, the agent receives a reward of $11$ due to health benefits. With bias $b = 2$, the agent initially believes this task is worth completing, but due to his bias, abandons the task at vertex $v$, after having already purchased the membership. This same example works more generally for project abandonment, where $(s, v)$ could represent the easier planning/conceptual stage while $(v, t)$ represents the more difficult execution.}
    \label{fig:gym}
\end{figure}

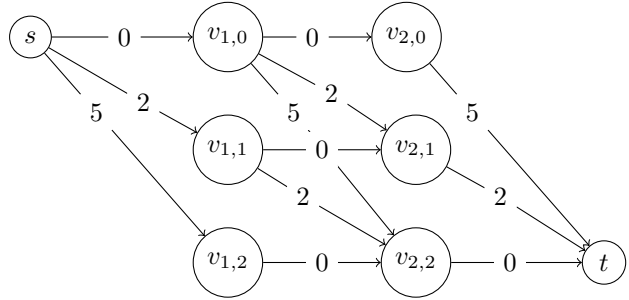
\begin{figure}[]
    \centering
    \begin{tikzpicture}
	\begin{pgfonlayer}{nodelayer}
		\node [style=standard] (0) at (-0.25, 6) {$s$};
		\node [style=label] (1) at (2.25, 6) {$0$};
		\node [style=standard] (2) at (9.75, 6) {$v_{2,0}$};
		\node [style=standard] (3) at (5, 6) {$v_{1,0}$};
		\node [style=label] (4) at (7.25, 6) {$0$};
		\node [style=standard] (5) at (5, 0) {$v_{1,2}$};
		\node [style=label] (6) at (7.5, 0) {$0$};
		\node [style=standard] (7) at (15, 0) {$t$};
		\node [style=standard] (8) at (10, 0) {$v_{2,2}$};
		\node [style=label] (9) at (12.5, 0) {$0$};
		\node [style=standard] (10) at (5, 3) {$v_{1,1}$};
		\node [style=label] (11) at (7.5, 3) {$0$};
		\node [style=standard] (13) at (10, 3) {$v_{2,1}$};
		\node [style=label] (14) at (2.75, 4.25) {$2$};
		\node [style=label] (15) at (7.75, 4.5) {$2$};
		\node [style=label] (16) at (7, 1.75) {$2$};
		\node [style=label] (17) at (1.5, 4) {$5$};
		\node [style=label] (18) at (12.25, 1.75) {$2$};
		\node [style=label] (19) at (11.5, 4) {$5$};
		\node [style=label] (20) at (6.75, 4) {$5$};
	\end{pgfonlayer}
	\begin{pgfonlayer}{edgelayer}
		\draw (0) to (1);
		\draw [style=weighted directed] (1) to (3);
		\draw (3) to (4);
		\draw [style=weighted directed, in=180, out=0] (4) to (2);
		\draw (5) to (6);
		\draw [style=weighted directed, in=180, out=0] (6) to (8);
		\draw (8) to (9);
		\draw [style=weighted directed] (9) to (7);
		\draw (10) to (11);
		\draw [style=weighted directed, in=180, out=0] (11) to (13);
		\draw (0) to (14);
		\draw [style=weighted directed] (14) to (10);
		\draw (3) to (15);
		\draw [style=weighted directed] (15) to (13);
		\draw (2) to (19);
		\draw [style=weighted directed] (19) to (7);
		\draw (13) to (18);
		\draw [style=weighted directed] (18) to (7);
		\draw (10) to (16);
		\draw [style=weighted directed] (16) to (8);
		\draw (0) to (17);
		\draw [style=weighted directed] (17) to (5);
		\draw (3) to (20);
		\draw [style=weighted directed] (20) to (8);
	\end{pgfonlayer}
\end{tikzpicture}
    \caption{\citet{roughgarden2016cs269i} gave this example to represent a three week course with two assignments. Let $v_{i,j}$ correspond to being in week $i$ with $j$ assignments completed, and suppose the reward of the course is $9$. Note that a student with bias $b = 2$ will procrastinate until $v_{2,0}$ and then ultimately abandon the course. However, if the instructor required the first assignment to be completed by week 2 (which corresponds to removing $v_{2,0}$), the student would complete the course.}
    \label{fig:class}
\end{figure}

The above results all focus on \textit{accommodating} present bias rather than \textit{alleviating} it. By that, we mean that the approaches all focus on whether the agent can be convinced to complete the task -- via edge deletion or reward dispersal -- but not on guarding the agent from suboptimal behavior induced by their bias. \cite{kleinberg2016sophisticated} partially investigates the latter question in a model involving \textit{sophisticated} agents, who plan around their present bias. They consider several types of \textit{commitment devices} -- tools by which sophisticated agents can constrain their future selves. However, these tools may require more powerful agents or designers and don't necessarily make sense for naive agents. We take a different approach -- we show that adding competition can simultaneously explain why present-biased agents may not perform exponentially poorly in ``natural'' games and guide task designers in encouraging biased agents towards optimal behavior. 

\subsection{Our Model}
In our model, a task is still represented by a directed, acyclic graph $G$, with a designated source $s$ and sink $t$. There are two naive present-biased agents, $A_1$ and $A_2$, both with bias $b$, who compete to get to $t$ first. The cost of a path is the sum of the weights along the path, and time is represented by the number of edges in the path, which we call the {\em length} of the path. In other words, each edge represents one unit of time. The first agent to get to $t$ gets a reward of $r$; ties are resolved by evenly splitting the reward. Recall that naive agents believe that they will behave optimally in the future. Thus, an agent currently at $u$ considers the cost to reach the target $t$ to be  $b c(u, v)$ plus the cost of the optimal path from $v$ to $t$ minus the reward of that path. More formally, let $\mathcal{P}(v \to t)$ denote the set of paths from $v \to t$ and let $P(s \to u)$ denote the path the agent has taken to $u$. Let $C_n(u, v)$ denote the remaining cost that the naive agent believes they will incur while at $u$ and planning to go to $v$. The subscript $n$ stands for ``naive'' (to help distinguish from $c(u, v)$, the cost of the edge $(u, v)$).  Then:
\begin{equation}
\label{eq:Cn}
    C_n(u, v) = b \cdot c(u, v) + \min_{\mathclap{P(v \to t) \in \sP(v \to t)}} c(P(v \to t)) - R_{A_2}(P(s \to u) \cup (u, v) \cup P(v \to t)),
\end{equation}
where $c(P) = \sum_{e \in P} c(e)$ denotes the cost of path $P$ and $R_{A_2}(P)$ denotes the reward of taking path $P$ from $s$ to $t$. This reward depends on the path the other agent $A_2$ takes. Specifically, 
if $A_2$ takes a path of length $k$, and $P(s \to u) \cup (u, v) \cup P(v \to t)$ is a path of length  $\ell$, then 
$$R_{A_2}(P(s \to u) \cup (u, v) \cup P(v \to t)) = \begin{cases}
r & \ell < k \\
\frac{r}{2} & \ell = k \\
0 & \text{ otherwise}.
\end{cases}
$$
We will often rewrite the second term in \eqref{eq:Cn}, for ease of notation, as $\min_{P_v} c(P_v) - R_{A_2}(P_{s\to u, v \to t})$. We sometimes refer instead to the naive agent's \textit{utility}, which is the negation of this cost.
Given this cost function, the naive agent chooses the successor of node $u$ via $$S(u) = \argmin_{v: (u,v) \in E} C_n(u, v)$$ See \autoref{fig:compEx} for an example. 

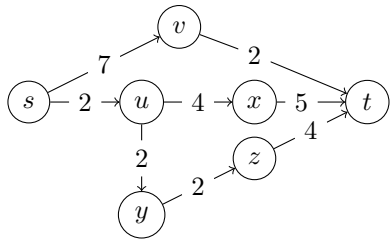
\begin{figure}[b]
    \centering
    \begin{tikzpicture}
	\begin{pgfonlayer}{nodelayer}
		\node [style=none] (0) at (0, 0) {};
		\node [style=standard] (1) at (0, 0) {$s$};
		\node [style=standard] (3) at (3, 0) {$u$};
		\node [style=standard] (4) at (6, 0) {$x$};
		\node [style=standard] (5) at (9, 0) {$t$};
		\node [style=standard] (6) at (3, -3) {$y$};
		\node [style=standard] (7) at (4, 2) {$v$};
		\node [style=standard] (8) at (6, -1.5) {$z$};
		\node [style=label] (9) at (1.5, 0) {$2$};
		\node [style=label] (10) at (3, -1.5) {$2$};
		\node [style=label] (11) at (4.5, -2.25) {$2$};
		\node [style=label] (12) at (7.5, -0.75) {$4$};
		\node [style=label] (13) at (7.25, 0) {$5$};
		\node [style=label] (14) at (2, 1) {$7$};
		\node [style=label] (15) at (4.5, 0) {$4$};
		\node [style=label] (16) at (6, 1.25) {$2$};
	\end{pgfonlayer}
	\begin{pgfonlayer}{edgelayer}
		\draw (1) to (9);
		\draw (3) to (10);
		\draw (6) to (11);
		\draw (8) to (12);
		\draw (1) to (14);
		\draw (4) to (13);
		\draw [style=weighted directed] (13) to (5);
		\draw [style=weighted directed] (12) to (5);
		\draw [style=weighted directed] (11) to (8);
		\draw [style=weighted directed] (10) to (6);
		\draw [style=weighted directed] (9) to (3);
		\draw [style=weighted directed] (14) to (7);
		\draw (7) to (16);
		\draw [style=weighted directed] (16) to (5);
		\draw (3) to (15);
		\draw [style=weighted directed] (15) to (4);
	\end{pgfonlayer}
\end{tikzpicture}
    \caption{Suppose $r = 5$, the bias $b = 2$, and assume $A_2$ takes path $(s, u, x, t)$. Then at $s$, $A_1$ prefers to take $u$ for perceived cost $4+4+5-2.5=10.5$. Notice that, due to the reward, the path $A_1$ believes he will take from $u$ is $(u, x, t)$, despite $(u, y, z, t)$ having lower cost. However, at $u$, $A_1$ evaluates the lower path to be cheaper, despite losing the race. This shows that a reward of $5$ does not ensure a Nash equilibrium on $(s, u, x, t)$ when $b=2$.}
    \label{fig:compEx}
\end{figure}

We now consider how this model of competition might both explain the outcomes of natural games and inform task designers on how to elicit better behavior from biased agents. For a natural game, consider the classic example of two companies competing to expand into a new market. Both companies want to launch a similar product, and are thus considering the same task graph $G$. The companies are also present biased, since shareholders often prefer immediate profit maximization/loss minimization over long term optimal behavior. The first company to enter the market gains an insurmountable advantage, represented by reward $r$. If the companies both enter the market at the same time, they split the market share, each getting reward $r/2$. This arrangement can be modeled within our framework, and the competition between the companies should lead them to play a set of equilibrium strategies.

For a designed game, consider the problem of encouraging students to submit final projects before they are due. The instructor sets  a deadline near the end of finals week to give students flexibility to complete the project when it best fits their schedule. The instructor also knows that (1) students tend to procrastinate and (2) trying to complete the final project in a few days is much more challenging than spreading it out. They would like  to convince students to work on and possibly submit their assignments early, \textit{without} changing the deadline (to allow flexibility for the students whom it suits best). One possible solution would be to give a small amount of extra credit to the first submission. How might they set this reward to encourage early submissions?

For another example of a designed game, consider a gym that enjoys increased enrollment at the start of the year. However, their new customers tend to not visit the gym frequently, and eventually cancel their membership. To remedy this, they try a new promotion, where new customers are offered free fitness classes, and the first customer(s) to attend 5 such classes gets a discount on their annual membership. How effective might such a scheme be at encouraging their new customers to regularly use their gym?
Across these three examples, the intuition is that competition will alleviate the harms of present bias by driving agents towards optimal behavior.

\subsection{Summary of Results}
We have introduced a model of competition for completing tasks along some graph. We warm up by analyzing these games
absent present bias. 
Namely, we classify all Nash equilibria for an arbitrary task graph with unbiased agents.  This analysis involves characterizing the set of paths which are the cheapest of a given length.

We then analyze these games when agents have equal present bias. We show that a very small reward induces a Nash equilibrium on the optimal path, for any graph with a \textit{dominant} path. This is a substantial improvement over the exponential worst case cost ratio experienced without competition. We then discuss how time-inconsistency defeats the intuition that higher rewards cause agents to prefer quicker paths. Despite this complication, we provide an algorithm that, given arbitrary graph $G$ and path $Q$, determines the minimum reward needed to get a Nash equilibrium on $Q$, if possible. 

Finally, we add an element of \textit{bias uncertainty} to the model, by drawing agents' biases iid from distribution $F$ and, for the $n$-fan, describe the relationship between $F$ and the reward required for a Bayes-Nash equilibrium on the optimal path. For a wide range of distributions, we find small rewards suffice to ensure that agents behave optimally (with high probability) in equilibrium. For the stronger goal of ensuring a constant expected cost ratio, it suffices to offer reward linear in $n$ when $F$ is not heavy-tailed; competition thus helps here as well.

\section{Nash Equilibria with Unbiased Competitors}
To build intuition, we first describe the Nash equilibria of these games when agents have no present bias. We also pinpoint where the analysis will change with the introduction of bias. Notice that each path $P$ in the graph is a strategy, with payoffs either $u_w = r - c(P)$, $u_t = r/2 - c(P)$ or  $u_l = -c(P)$, depending on whether the agent wins, ties or loses, respectively. (These in turn depend on the path taken by the opponent.) We first rule out dominated paths. Notice that if $u_l(P) \ge u_w(P')$, path $P'$ is dominated by path $P$, regardless of the path taken by the opponent. Also, if $u_w(P) \ge u_w(P')$ and $|P| \le |P'|$ (where $|P|$ is the number of edges in $P$), then $P'$ is dominated. Therefore, for any length $k$, a single cheapest path of length $k$ will (weakly) dominate all other paths of length $k$.

For a given graph $G$, let $P_1, \ldots, P_n$ be a minimal set of non-dominated paths, where  $|P_i| < |P_{i+1}|$ for each $1\le i < n$.  Thus, $P_1$ is the \textit{quickest} path, the remaining path of minimum length. Summarizing what we know about these paths:
\begin{enumerate}
    \item \textit{Winning is better than losing}: for any pair of paths $(P_i, P_j)$, we know that $u_w(P_i) > u_l(P_j)$. Thus, in particular, $c(P_1) - c(P_n) \le r$.
    \item \textit{Longer paths are more rewarding}: That is, $c(P_i)> c(P_{i+1})$ for each $i$. Otherwise, $P_{i+1}$ would be dominated since its length is greater. Therefore, in particular,  $P_n$ is the \textit{cheapest} path, i.e., the lowest cost/weight path from $s$ to $t$.
\end{enumerate}
We're interested in characterizing, across all possible task graphs, the pure Nash equilibria, restricting attention to paths in $P_1, \ldots, P_n$. 

\begin{proposition}
    Let  $G$ be an arbitrary task graph. As above, let $P_1, \dots, P_n$ be a minimal set of (non-dominated) paths ordered  so $P_1$ is the quickest and $P_n$ the cheapest. Suppose $n \ge 3$. Then, path $P_i$, where $i > 1$, is a symmetric Nash equilibrium if 
    and only if $c(P_{i-1}) - c(P_i) \ge r/2$. $P_1$ is a symmetric Nash if and only if 
    $ c(P_1)- c(P_n) \le r/2$. There are no other pure Nash equilibria. Therefore, there can be either 0, 1 or 2 pure Nash equilibria.

    If $n = 2$, there is an additional asymmetric pure Nash equilibrium where one player plays the quickest path $P_1$ and the other plays the cheapest path $P_2$ if 
    $c(P_1) - c(P_2) = r/2$.
\end{proposition}
\begin{proof}
    First, suppose $n \ge 3$ and the opponent picks $P_i$, where $i>1$. If the agent plays $P_{i-1}$, they get $u_w(P_{i-1})$. Since winning is always better than losing, the agent need not consider any $P_{> i}$. Similarly, since longer paths are cheaper, the agent will not take any $P_{< i-1}$. Thus, the only choices are between tying on $P_i$ or winning on $P_{i-1}$, so $P_i$ is a symmetric Nash exactly when $u_t(P_i) \ge u_w(P_{i-1})$, or equivalently $c(P_{i-1}) - c(P_i) \ge r/2$. Similarly, if the opponent plays $P_1$, the agent chooses between tying on $P_1$ or losing; in the latter case, losing on the cheapest path $P_n$ gives highest utility. The agents therefore have a symmetric Nash exactly when tying on $P_1$ is better than losing on $P_n$, or  $ c(P_1)- c(P_n) \le r/2$.

    The only strategy profile we haven't considered is $(P_{i-1}, P_i)$; when might this be a Nash equilibrium? Suppose the opponent plays $P_{i-1}$. Either $i -1 > 1$, and so the best response is $P_{i-2}$ or $P_{i-1}$, or $i-1=1$, and so the best response is $P_{1}$ or $P_n$. If $n \ge 3$, then $P_{2} \neq P_n$ and so $(P_{i-1}, P_i)$ is not an equilibrium. Otherwise, if $n = 2$, then $P_2 = n$ and so we get an asymmetric Nash when tying is dominated, i.e. $u_t(P_1) \le u_l(P_2)$ and $u_w(P_1) \ge u_t(P_2)$. This is the same as  $c(P_1) - c(P_2) = r/2$.
\qed
\end{proof}
If we apply the same reasoning when the tie-breaking rule is different, for example, when ties are resolved by both players getting the full reward, any $P_i$ is a (symmetric) pure Nash (and there are no other pure Nash). On the other hand, if ties result in neither player receiving any reward, the same reasoning implies that we only have a pure (asymmetric) Nash if $n=2$; if $n \ge 3$, there are no pure Nash.

We next turn our attention to the biased version of this problem. In the unbiased case, we could take a ``global'' view of the graph, and think about paths purely in terms of their overall length and cost. But when agents are biased, the actual structure of the path is very important; time-inconsistency means that agents look at paths \textit{locally}, not globally. It is thus very difficult to cleanly rule out dominated paths -- even paths with exponentially high cost may be taken, as we see next.

\section{Nash Equilibria to Elicit Optimal Behavior from Biased Agents}
We now assume that the agents are both naive, present biased agents, with shared bias parameter $b$.\footnote{ The homogenity of the agents is not particularly important to our results in this section. If the agents have different bias parameters, our results go through by setting $b$ equal to the larger of the two biases.} Our high-level goal is to show that competition convinces biased agents to take cheap paths, as unbiased agents do without competition. To this end, we show that a small amount of reward creates a Nash equilibrium on the cheapest path, first on the $n$-fan, and then for all graphs which have a \textit{dominant} path -- a cheapest path that is also the \textit{uniquely} quickest path.

\subsection{The \texorpdfstring{$n$}{n}-fan}
We focus first on the $n$-fan since all graphs with exponential cost ratio for naive agents have a large $k$-fan as a graph minor \citep{kleinberg2014naive}. For an example of the $n$-fan, refer back to \autoref{fig:n-fan}; note in particular that we assume the bias parameter $b > c > 1$. We show that, with a very small amount of reward, competition encourages optimal behavior in naive agents (and further, we fully characterize \textit{all} equilibria on the $n$-fan). To aid our proofs and discussions regarding the $n$-fan, we define $P_i$ as the path from $s$ to $t$ containing edge $(v_i, t)$. Define $P_0$ as the direct path $(s, t)$.
\begin{theorem}
    Let $G$ be an $n$-fan, and for any bias $b > c$, define $\eps = b - c$. Then, with a reward of $r \ge 2\eps$, there will be a Nash equilibrium on the optimal path, for two agents with bias $b$. Further, if $r \le 2 \eps c^{n-1}$, there is also a Nash equilibrium on the longest path. There are no other pure Nash equilibria, for any value of $r$.
\end{theorem}
\begin{proof}
    We first show that $r \ge 2\eps$ guarantees the existence (not uniqueness) of an optimal pure Nash equilibrium. Suppose $A_2$ takes the optimal path $P_0$. While standing at $s$, $A_1$ evaluates the cost to $t$ as $b - r/2$, i.e. $C_n(s, t) = b-r/2$. Similarly, $C_n(s, v_1) = c$, so with $r \ge 2\eps$, $A_1$ (weakly) prefers to go directly to $t$. Using a similar calculation, we can see that when $r \le 2 \eps c^{n-1}$, there is a Nash equilibrium on the longest path, $P_n$. This result is somewhat surprising -- even with an extremely high reward, the naive agents may take the longest path. To see why, note that at any vertex $v_i$, the naive agent is only comparing two options, going to $t$ right now or waiting one step and going to $t$ (since they naively believe they will behave optimally in the future). Further, when $A_2$ takes the longest path, $A_1$ gets the same reward from $(v_{i}, t)$ and $(v_{i+1}, t)$. Thus, just as the agent prefers to procrastinate in the $n$-fan, they similarly procrastinate here until $v_{n-1}$, where the reward must be very large for them to deviate from the longest path.

    To show that there are no other Nash equilibria, suppose $A_2$ takes $P_i$, where $1 \le i < n$. As we explained above, since $A_1$ gets the same \textbf{reward} from paths $P_1$ to $P_{i-1}$ and the same reward from paths $P_{i+1}$ to $P_n$, the agent will always procrastinate to at least $v_{i-1}$, and if he chooses to go past $v_i$, he will procrastinate until $v_n$. Now, suppose $A_1$ is standing at $v_{i-1}$. He will take path $P_{i-1}$ if $bc^{i-1} - r < c^i - r/2$, which implies that $r > 2\eps c^{i-1}$.
    Either the reward satisfies this, which shows that $P_i$ is not a Nash equilibrium, or $r \le 2\eps c^{i-1}$, and $A_1$ continues to $v_i$. Suppose he continues to $v_i$. Then, he (weakly) prefers $P_{i}$ if $ bc^i - r/2 \le c^{i+1}$, which implies that $r \ge 2\eps c^i$.
    However, $r \le 2\eps c^{i-1} < 2\eps c^i$, so $r$ does not satisfy this. Thus, $A_1$ prefers to procrastinate until $v_n$. In summary, if $A_2$ takes path $P_i$, $A_1$ either prefers $P_{i-1}$ or $P_n$ (depending on the reward). The same argument applies to $A_2$, so there are no Nash equilibria on paths $P_1, \dots, P_{n-1}$, regardless of $r$.
\qed
\end{proof}
The ability to enable optimal behavior with just a constant reward stands in striking contrast to the setting where agents can abandon the task at any time. With neither competition nor internal rewards, the designer would have to place an exponentially high reward on $t$ (i.e. $r\ge c^n$) for biased agents to finish the task.\footnote{However, with internal \textit{edge} rewards, the designer could spend the same total reward to motivate two agents by directly adding a reward to the $(s, t)$ edge. This equivalence doesn't hold in all graphs, as we show in the next section.}

\subsection{Graphs with an Unbiased Dominant Strategy}
We now generalize our results to a larger category of graphs.   To focus exclusively on the irrationality of present bias rather than the optimization problem of choosing between cheap, long paths and short, expensive paths, we focus on graphs with a \textit{dominant path} -- a cheapest path\footnote{There may be other cheapest paths which are longer.} that is also the \textit{uniquely} quickest path. This is the case in the $n$-fan. In this setting, the problem is trivial for unbiased agents; simply take this dominant path. But for biased agents, the problem is still interesting; as the example of the $n$-fan shows, they may take paths that are exponentially more costly than the dominant path. However, we prove that a small amount of competition and reward suffices to ensure the existence of a Nash equilibrium where both agents take the dominant path.

\begin{theorem}\label{thm:neDom}
    Suppose $G$ is a task graph that has a \emph{dominant} path, $O$. Then, a reward of $r \ge 2b \cdot \max_{e \in O}c(e)$ guarantees a Nash equilibrium on $O$, for two agents with bias $b$.  
\end{theorem}
\begin{proof}
    Assume that $A_2$ takes $O$. Recall that a biased agent perceives the remaining traversal cost of going from $v$ to $t$ as:
    \begin{align*}
        C_n(u, v) = b \cdot c(u, v) + \min_{P_v} c(P_v) - R_{A_2}(P_{s\to u, v\to t})
    \end{align*}
    We know that for any vertex $v^*$ on the dominant path, the path that minimizes the second term is just the fragment of the dominant path from $v^* \to t$ (it is both the quickest and cheapest way to get from $v^*$ to $t$). Further, any deviation from the dominant path results in no reward. So, for any $v$ not on the dominant path, the path that minimizes the second term is again the cheapest path from $v \to t$. Thus, the cost equation simplifies to:
    \begin{align*}
        C_n(u, v) = b \cdot c(u, v) + d(v) - r/2\cdot \mathbf{1}\{D\},
    \end{align*}
    where $d(v)$, the \textit{distance} from $v \to t$, denotes the cost of the cheapest path from $v$ to $t$ (ignoring rewards) and $\mathbf{1}\{D\}$ is simply an indicator variable that's $1$ if the agent has not deviated from the dominant path. 
    
    Now, let $O = (s=v_0^*, v_1^*, v_2^*, \dots, t=v_l^*)$ be the dominant path and suppose $A_2$ takes this path. In order for $A_1$ to choose $O$, we require, for all $i$:
    \begin{align*}
        &S(v_i^*) = v_{i+1}^* \\
        \iff& v_{i+1}^* = \argmin_{v: (v_i^*, v) \in E} b \cdot c(v_i^*, v) + d(v) - r/2\cdot \mathbf{1}\{D\} \\
        \iff& \forall v: (v_i^*, v) \in E, bc(v_i^*, v) + d(v) \ge bc(v_i^*, v_{i+1}^*) + d(v_{i+1}^*) - r/2
    \end{align*}
    Now, let $i$ be arbitrary, let $v \neq v_i^*$ be an arbitrary neighbor of $v_i^*$, and for ease of notation, let $c = c(v_i^*, v), c^* = c(v_i^*, v_{i+1}^*), d = d(v)$, and $d^* = d(v_{i+1}^*)$. Then, we get the following bound on the reward: $r/2 \ge b(c^* - c) + d^* - d$.
    To get a rough sufficient bound, notice that $c+d \ge c^* + d^*$, since $O$ is the cheapest path. This implies that $bc^* > b(c^*-c) + d^* - d$. Thus, it suffices to set $r \ge 2bc^*$ in order to ensure $S(v_{i}^*) = v_{i+1}^*$. Repeating this argument for all $i$, we see that a sufficient reward is $r = 2b \cdot \max_{e \in O} c(e)$. 
\qed
\end{proof}
One might object to our claim that $r = 2b \max_{e \in O} c(e)$ is ``small''. To calibrate our expectations, notice that we can view this problem as an agent trying to pick between several options (i.e. paths), each with their own reward and cost structure. We want to convince the agent to pick one particular option -- namely, the cheapest path. But it would be unreasonable to expect that a reward significantly smaller than the cost of \textit{any} option would sway the agent's decision. Our theorem above shows that a reward that is at most proportional to the \textit{cheapest} option suffices; and in many cases the reward is only a fraction of the cost of the optimal path (e.g. when the optimal path has balanced cost among many edges).

For a point of comparison, even internal edge rewards (which required the same reward budget as competitive rewards for the $n$-fan) can require $O(n)$ times as much reward in some instances. To see the intuition, notice that internal edge rewards must be applied at every step where the agent might want to deviate. The agent also immediately ``consumes'' this reward; it doesn't impact his future decision making. However, the competitive reward is ``at stake'' whenever the agent considers deviating; this reward can sway the agent's behavior without being immediately consumed. For a concrete example, see \autoref{fig:comp_less_reward}.
\begin{figure}[t]
    \centering
    \begin{tikzpicture}
	\begin{pgfonlayer}{nodelayer}
		\node [style=standard] (0) at (0, 0) {$s$};
		\node [style=label] (1) at (2.5, 0) {$1$};
		\node [style=standard] (3) at (5, 0) {$v_1$};
		\node [style=label] (4) at (7.5, 0) {$1$};
		\node [style=standard] (5) at (5, 4) {$u_1$};
		\node [style=label] (6) at (7.5, 2) {$2$};
		\node [style=label] (7) at (5, 2) {$0$};
		\node [style=standard] (9) at (10, 0) {$v_2$};
		\node [style=label] (10) at (12.5, 0) {$1$};
		\node [style=standard] (11) at (10, 4) {$u_2$};
		\node [style=label] (12) at (12.5, 2) {$2$};
		\node [style=label] (13) at (10, 2) {$0$};
		\node [style=standard] (15) at (15, 0) {$v_3$};
		\node [style=label] (16) at (17.5, 0) {$1$};
		\node [style=label] (18) at (20, 0) {$\cdots$};
		\node [style=label] (19) at (22.5, 0) {$1$};
		\node [style=standard] (20) at (25, 0) {$v_n$};
		\node [style=label] (21) at (27.5, 0) {$1$};
		\node [style=standard] (22) at (25, 4) {$u_n$};
		\node [style=label] (23) at (27.5, 2) {$2$};
		\node [style=label] (24) at (25, 2) {$0$};
		\node [style=standard] (25) at (30, 0) {$t$};
	\end{pgfonlayer}
	\begin{pgfonlayer}{edgelayer}
		\draw (0) to (1);
		\draw [style=weighted directed, in=180, out=0] (1) to (3);
		\draw (3) to (4);
		\draw (3) to (7);
		\draw [style=weighted directed] (7) to (5);
		\draw (5) to (6);
		\draw (9) to (10);
		\draw (9) to (13);
		\draw [style=weighted directed] (13) to (11);
		\draw (11) to (12);
		\draw [style=weighted directed] (4) to (9);
		\draw [style=weighted directed] (6) to (9);
		\draw (15) to (16);
		\draw [style=weighted directed] (10) to (15);
		\draw [style=weighted directed] (12) to (15);
		\draw (18) to (19);
		\draw [style=weighted directed] (16) to (18);
		\draw (20) to (21);
		\draw (20) to (24);
		\draw [style=weighted directed] (24) to (22);
		\draw (22) to (23);
		\draw [style=weighted directed] (21) to (25);
		\draw [style=weighted directed] (23) to (25);
		\draw [style=weighted directed] (19) to (20);
	\end{pgfonlayer}
\end{tikzpicture}
    \caption{A graph with many suboptimal deviations. For an agent with bias $b > 2$, a designer with access to only edge rewards must spend $O(n)$ total reward for optimal behavior ($b-2$ on each $(v_i, v_{i+1})$ edge). In our competitive setting, only $2(b-2)$ total reward is required.}
    \label{fig:comp_less_reward}
\end{figure}
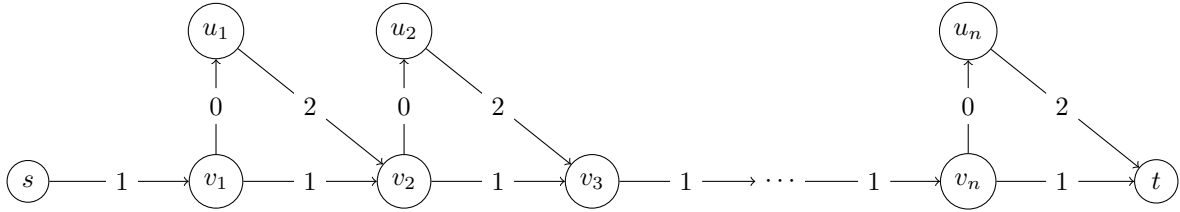

\subsection{Increasing the Number of Competitors}

A very natural extension to this model would involve more than 2 agents competing. The winner takes all, and ties are split evenly among those who tied. However, this modification doesn't change much \emph{when trying to get a Nash equilibrium on the dominant path}. The only change is that if $m$ agents are competing, the reward needed is $O(m)$, as a single agent will get a $1/m$ fraction of the reward in a symmetric equilibrium. This is true because there is no way for any agent to beat the dominant path, and claim the entire $O(m)$ reward for themselves. So if the reward is scaled appropriately (i.e. in \autoref{thm:neDom} set $r \ge mb \cdot \max_{e \in O} c(e)$), we will still guarantee a Nash equilibrium. Put another way, the \textit{per-agent} reward needed for a Nash equilibrium on the dominant path does not change as the number of competitors varies.

\section{General Nash Equilibria}
In this section, we describe a polynomial time algorithm that, given an arbitrary graph $G$, path $Q$ and bias $b$, determines if $Q$ can be made a Nash equilibrium, and if so, the minimum required reward to do so. Finding and using this minimum required reward will generally cost much less than  the bound given by \autoref{thm:neDom}. Moreover, this algorithm does not assume the existence of a dominant path. We start by describing how time-inconsistency defeats the intuition that higher rewards cause agents to prefer quicker paths. We then present an algorithm which computes the minimum $r$ that results in $Q$ being a symmetric Nash equilibrium.\footnote{In fact, the algorithm will return a set of $O(n)$ disjoint intervals containing \textit{all} values of $r$ that result in $Q$ being a symmetric Nash equilibrium.}

\subsection{Higher Rewards Need Not Encourage Quicker Paths}
The proof of  \autoref{thm:neDom} suggests the following algorithm for this problem. Start with a reward of $0$, and step along each vertex $u \in Q$, increasing the reward by just enough to ensure $A_1$ stays on $Q$ for one more step (assuming $A_2$ is taking $Q$). 
However, if $A_1$ wants to deviate onto a quicker/tied path at any point, return $\bot$; decreasing the reward would cause them to deviate earlier, and, intuitively, it seems that increasing the reward could not cause them to switch back to $Q$ from the quicker/tied path. 
After one pass, simply pass through again to ensure that the final reward doesn't cause $A_1$ to deviate early on. The following lemma shows that this algorithm is tractable, by showing that we can compute the minimum reward required for $A_1$ to stay on $Q$ at any step (and determine whether $A_1$ wants to deviate onto a quicker path).

\begin{lemma}\label{lm:bf}
    Assuming that $A_2$ takes path $Q$, $A_1$ can efficiently compute $\min_{P_v} c(P_v) - R_{A_2}(P_{s\to u,v \to t})$ by considering the cheapest path (from $v\to t$), the cheapest path where $A_1$ ties $A_2$, and the cheapest path where $A_1$ beats $A_2$. (Some of these paths may coincide, and at least one must exist).
\end{lemma}
\begin{proof}
    If $A_1$ assumes that $A_2$ takes path $Q$, then the reward is simply $r$ if the path under consideration wins, $r/2$ if it ties, and $0$ if it loses. Thus, the minimizer must be the cheapest paths in one of these three cases (and the cheapest $v\to t$ path always exists; the other two may not). Specifically, let $k$ be the number of edges from $v \to t$ on $Q$, and let $c_\infty, c_{\le k}, c_{< k}$ denote the cost of the cheapest $v \to t$ paths with any number of edges, at most $k$ edges, and strictly less than $k$ edges respectively. Then, $\min_{P_v} c(P_v) - R_{A_2}(P_{s\to u,v \to t}) = \min(c_\infty, c_{\le k} - r/2, c_{< k} - r)$. These paths can be efficiently computed using the Bellman-Ford algorithm to determine the cheapest path with at most $k$ edges for any $k$.
\qed
\end{proof}
Unfortunately, while tractable, the approach described above does not yield a correct algorithm. This is because it relies implicitly on the following two properties, which formalize the intuition that increasing the reward causes agents to favor quicker paths.
\begin{property}\label{prop1}
    If a reward $r$ guarantees a Nash equilibrium on some path $Q$, any reward $r' > r$ will either (a), still result in a Nash equilibrium on $Q$, or (b), cause an agent to deviate to a \textit{quicker} path $Q'$.
\end{property}
\begin{property}\label{prop2}
    If $A_1$ deviates from $Q$ onto a quicker/tied path for some reward $r$, increasing the reward will not cause them to follow $Q$.
\end{property}
Both properties are intuitive -- if we increase the reward, that should motivate the agent to take a path that beats their opponent. And vice versa -- increasing the reward shouldn't cause them to shift onto a slower path or shift between equal length paths. But surprisingly, both properties are false, as the following examples demonstrate.
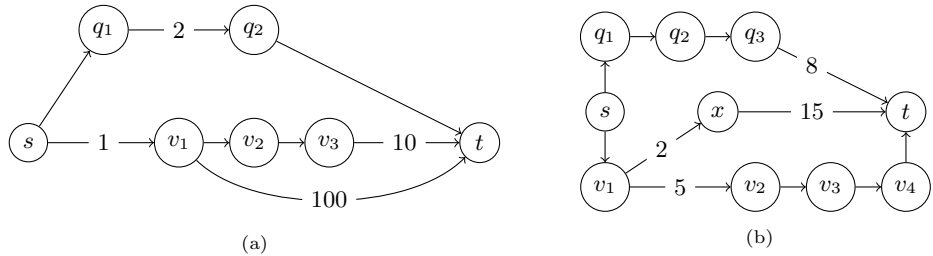
\begin{figure}[H]
    \centering
    \begin{subfigure}{0.4\textwidth}
        \centering
        \begin{tikzpicture}
	\begin{pgfonlayer}{nodelayer}
		\node [style=standard] (0) at (0, 0) {$s$};
		\node [style=label] (1) at (2, 0) {$1$};
		\node [style=standard] (2) at (4, 0) {$v_1$};
		\node [style=standard] (4) at (2, 3) {$q_1$};
		\node [style=standard] (7) at (6, 0) {$v_2$};
		\node [style=standard] (9) at (6, 3) {$q_2$};
		\node [style=standard] (12) at (8, 0) {$v_3$};
		\node [style=label] (13) at (10, 0) {$10$};
		\node [style=standard] (14) at (12, 0) {$t$};
		\node [style=label] (16) at (4, 3) {$2$};
		\node [style=label] (18) at (8, -1.5) {$100$};
	\end{pgfonlayer}
	\begin{pgfonlayer}{edgelayer}
		\draw (0) to (1);
		\draw [style=weighted directed, in=180, out=0] (1) to (2);
		\draw (12) to (13);
		\draw [style=weighted directed] (13) to (14);
		\draw (4) to (16);
		\draw [style=weighted directed] (16) to (9);
		\draw [style=weighted directed, bend right=45, looseness=0.75] (2) to (14);
		\draw [style=weighted directed] (0) to (4);
		\draw [style=weighted directed] (2) to (7);
		\draw [style=weighted directed] (7) to (12);
		\draw [style=weighted directed] (9) to (14);
	\end{pgfonlayer}
\end{tikzpicture}
        \caption{}
    \end{subfigure}
    \begin{subfigure}{0.4\textwidth}
        \centering
        \begin{tikzpicture}
	\begin{pgfonlayer}{nodelayer}
		\node [style=standard] (0) at (0, 0) {$s$};
		\node [style=standard] (2) at (0, 2) {$q_1$};
		\node [style=standard] (4) at (0, -2) {$v_1$};
		\node [style=standard] (5) at (2, 2) {$q_2$};
		\node [style=standard] (7) at (4, -2) {$v_2$};
		\node [style=standard] (8) at (4, 2) {$q_3$};
		\node [style=label] (9) at (5.5, 1.25) {$8$};
		\node [style=standard] (10) at (8, 0) {$t$};
		\node [style=label] (12) at (2, -2) {$5$};
		\node [style=standard] (14) at (6, -2) {$v_3$};
		\node [style=standard] (16) at (8, -2) {$v_4$};
		\node [style=standard] (18) at (3, 0) {$x$};
		\node [style=label] (19) at (1.5, -1) {$2$};
		\node [style=label] (20) at (5.5, 0) {$15$};
	\end{pgfonlayer}
	\begin{pgfonlayer}{edgelayer}
		\draw (8) to (9);
		\draw [style=weighted directed] (9) to (10);
		\draw (4) to (12);
		\draw [style=weighted directed] (12) to (7);
		\draw [style=weighted directed] (4) to (18);
		\draw [style=weighted directed] (18) to (10);
		\draw [style=weighted directed] (0) to (2);
		\draw [style=weighted directed] (2) to (5);
		\draw [style=weighted directed] (5) to (8);
		\draw [style=weighted directed] (0) to (4);
		\draw [style=weighted directed] (7) to (14);
		\draw [style=weighted directed] (14) to (16);
		\draw [style=weighted directed] (16) to (10);
	\end{pgfonlayer}
\end{tikzpicture}
        \caption{}
    \end{subfigure}
    \caption{Graphs which do not exhibit the two expected properties. Unlabeled edges have cost $0$.}
    \label{fig:alg_ce}
\end{figure}
For \autoref{prop1}, consider the graph in \autoref{fig:alg_ce}(a) and define paths $Q = (s, q_1, q_2, t)$, $V = (s, v_1, v_2, v_3, t)$, and $X = (s, v_1, t)$. Suppose both agents have bias $10$ and that $A_2$ takes $Q$. Then a reward of $1$ guarantees that $A_1$ takes $Q$, as the optimal path from $v_1 \to t$ would follow $V$. However, if $r = 300$, the optimal path from $v_1 \to t$ follows $X$. So, $C_n(s, q_1) = -148 > C_n(s, v_1) = -200$ and so the agent goes to $v_1$.
But at $v_1$, the perceived cost of $(v_1, t)$ is actually $1000$, and so the agent prefers to take path $V$. Thus, increasing the reward from $1$ to $300$ causes the agent to deviate from $Q$ onto a slower path!

For \autoref{prop2}, consider the graph in \autoref{fig:alg_ce}(b) and define paths $Q, V$, and $X$ in the obvious manner. Again, suppose both agents have bias $10$ and that $A_2$ takes $Q$. Then, with a reward of $10$, the optimal path from $v_1 \to t$ would follow $V$. So, $C_n(s, v_1) = 5 > C_n(s, q_1) = 8 - 5 = 3$, and $A_1$ goes to $q_1$ and thus takes $Q$. If $r = 2$, the optimal path from $v_1 \to t$ still follows $V$. But now $C_n(s, v_1) = 5 < C_n(s, q_1) = 8-1 = 7$. Thus, the agent goes to $v_1$. And here, with $b = 10$, deviating to $X$ is more attractive than remaining on $V$, and thus the agent takes $X$. So, although $A_1$ deviates from $Q$ to a quicker path for reward $2$, they remain on $Q$ with reward $10$.

To summarize, \autoref{prop1} fails because present biased agents can take slower paths than they planned and \autoref{prop2} fails because present biased agents can take quicker paths than they planned. In other words, while higher rewards do \textit{tempt} agents to take quicker paths, and lower rewards tempt agents towards cheaper paths, their time inconsistency may make them do the opposite.

\subsection{The Algorithm}
We now present an algorithm which finds the set of rewards which induce a symmetric Nash equilibrium on a path $Q$ (\autoref{alg:ne}).
At a high level, the algorithm narrows down the set $\mathcal{I}^*$ of \textit{feasible} rewards (rewards that ensure that $Q$ is a Nash equilibria) by computing the set $\mathcal{I}_v$ of rewards that ensure that $A_1$ takes $(u, v)$ for every $(u, v) \in Q$. The key idea is that we can efficiently compute all $r$ that ensures that $A_1$ prefers $(u, v)$ over $(u, v')$ by splitting into cases based on whether the optimal paths from $v \to t$ and $v' \to t$ involve winning, tying, or losing. We now prove the correctness and runtime of the algorithm.
\begin{theorem}
    \autoref{alg:ne} returns the minimum $r$ that ensures that $Q$ is a Nash equilibrium, or $\bot$ if no such $r$ exists. Further, it runs in polynomial time.
\end{theorem}
\begin{proof}
    Assume that $A_2$ takes $Q$. The correctness of the algorithm follows from this lemma:
    \begin{lemma}
        Let $(u, v) \in Q$ and $(u, v' \neq v) \in G$ be arbitrary. $\mathcal{I}_{v,v'}$ contains the set of all rewards that ensure that $A_1$ prefers $(u, v) \in Q$ over $(u, v') \notin Q$ when standing at $u$.
    \end{lemma}
    \begin{proof}
        Suppose that the length of the path from $u$ to $t$ in $Q$ is $k$.
        From \autoref{lm:bf}, we know that $C_n(u, v) = \{c_\infty, c_{\le k} - r/2, c_{< k} - r\}$. Note that $c_{< k}$ represents the cheapest way to win, $c_{\le k}$ the cheapest way to win or tie, and $c_\infty$ the cheapest way to win, tie, or lose (i.e. the cheapest path). So, $c_\infty \le c_{\le k} \le c_{< k}$. $d_\infty, d_{\le k}$, and $d_{< k}$ are defined similarly with respect to $v'$. Without loss of generality, assume that all these values exist (i.e. that one \textit{can} win, tie, or lose from $v$ and $v'$); if this isn't the case, then all values which rely on them can be safely removed.

        Now, define $c^*(r) = \min(c_\infty, c_{\le k} - r/2, c_{< k} - r)$. This function $c^*(r)$ is simply the piecewise combination of three linear functions; $r_1 = 2(c_{\le k} - c_\infty)$ and  $r_2 = 2(c_{< k} - c_{\le k})$ represent the switching points. Define $d^*(r), s_1$ and $s_2$ similarly. Let $\mathcal{I}$ be the pairwise intersection\footnote{To obtain the pairwise intersection of two sets of intervals, simply take the Cartesian product of the two sets to obtain pairs, and then take the intersection within each pair to obtain a new set of intervals.} $\{[0,r_1], [r_1, r_2], [r_2, \infty]\} \times^\cap \{[0,s_1], [s_1, s_2], [s_2, \infty]\}$. Then, notice that for any $I \in \mathcal{I}$, $d^*(r) - c^*(r)$ has a closed form over $r \in I$. Since $A_1$ prefers $(u, v)$ over $(u, v')$ if $bc(u, v') + d^*(r) > bc(u, v) + c^*(r)$, this expression can be solved for $r$ in the interval $I$, yielding a single (sub)interval. Let $I_r$ denote this subinterval. Since $\mathcal{I}$ is a partition of all possible $r$, this means that $\mathcal{I}_{v, v'}$, the union of all intervals $I_r$, will contain all $r$ which ensure that $A_1$ prefers $(u, v)$ to $(u, v')$.
    \qed
    \end{proof}
    With this lemma, the proof of correctness is simple. $S(u) = v$ if and only if $r$ is in all $\mathcal{I}_{v, v'}$, which is what $\mathcal{I}_v$ tracks. And thus $A_1$ sticks with $Q$ if and only if $r$ is in all $\mathcal{I}_v$, which is exactly what $\mathcal{I}^*$ tracks.

    To prove that the algorithm is efficient, notice that $\mathcal{I}_{v, v'}$ contains at most $5$ intervals since $|\mathcal{I}| \le 5$. Further, the pairwise intersection of two sets of intervals can be computed in time $O(n+m)$ and is of size $\le n+m$, where $n$ and $m$ are the size of the input sets. Thus, all the interval intersections are efficient; in general, the runtime of the algorithm is dominated by the Bellman-Ford subroutine.
\qed
\end{proof}
\begin{algorithm}[bt]
    \KwIn{A DAG $G$, with source $s$ and sink $t$, path $Q$ in $G$, and bias factor $b$}
    \KwOut{The minimum reward $r$ for $Q$ to be a Nash equilibrium, or $\bot$ if not possible.}
    $k \gets \len(Q)$ \\
    $\mathcal{I}^* = \{[0, \infty]\}$ \\
    \ForEach{$(u,v) \in Q$}{
        $k \gets k-1$ \\
        $c_\infty, c_{\le k}, c_{< k} \gets \textsf{Bellman-Ford}(v, \infty), \textsf{Bellman-Ford}(v, k), \textsf{Bellman-Ford}(v, k-1)$ \\
        $r_1, r_2 \gets 2(c_{\le k} - c_\infty), 2(c_{< k} - c_{\le k})$\\
        $\mathcal{I}_v = \{[0, \infty]\}$ \\
        \ForEach{$(u, v' \neq v) \in G$} {
            $d_\infty, d_{\le k}, d_{< k} \gets \textsf{Bellman-Ford}(v', \infty), \textsf{Bellman-Ford}(v', k), \textsf{Bellman-Ford}(v', k-1)$ \\
            $s_1, s_2 \gets 2(d_{\le k} - d_\infty), 2(d_{< k} - d_{\le k})$ \\
            \tcp{$\times^\cap$ returns the pairwise intersection of two sets of intervals} 
            $\mathcal{I} \gets \{[0,r_1], [r_1, r_2], [r_2, \infty]\} \times^\cap \{[0,s_1], [s_1, s_2], [s_2, \infty]\}$ \\
            $\mathcal{I}_{v, v'} \gets \emptyset$\\
            \ForEach{$I \in \mathcal{I}$}{
                $c^* \gets \min(c_\infty, c_{\le k} - r/2, c_{< k} - r)$ assuming $r \in I$\\
                $d^* \gets \min(d_\infty, d_{\le k} - r/2, d_{< k} - r)$ assuming $r \in I$\\
                $I_r \gets$ subinterval of $I$ containing $r$ s.t. $d^* - c^* > b(c(u, v) - c(u,v'))$ \\
                $\mathcal{I}_{v, v'} \gets \mathcal{I}_{v, v'} \cup \{I_r\}$\\
            }
            $\mathcal{I}_v \gets \mathcal{I}_v \times^\cap \mathcal{I}_{v, v'}$
        }
        $\mathcal{I}^* \gets \mathcal{I}^* \times^\cap \mathcal{I}_v$
    }
    \eIf{$\mathcal{I}^* = \emptyset$}{
        \Return $\bot$
    } {
        \Return smallest $r \in \mathcal{I}^*$
    }
    \caption{Set Nash equilibrium on given path, if possible. Uses $\textsf{Bellman-Ford}(v, i)$ as a subroutine, to get the cost of the cheapest $v\to t$ path with at most $i$ edges. Also uses a subroutine to compute the pairwise intersection of two sets of disjoint intervals.}
    \label{alg:ne}
\end{algorithm}

\section{Extending the Model with Bias Uncertainty and Multiple Competitors}
One of the shortcomings of the prior results is that agents are assumed to have publicly known, identical biases. Neither identical nor known biases are very realistic assumptions. In general, it seems odd for the agents to know the exact present bias of their opponent, or for the agents to always have the same bias. From a design perspective, the task designer might not have full knowledge of the agents' biases,  and so may be unsure how to set the reward. We therefore consider this game played by two agents whose biases are uncertain. The agents' biases are represented by random variables $B_1$ and $B_2$ drawn iid from distribution $F$, which is publicly known to both the agents and the designer. $b_1$ and $b_2$ correspond to the realizations of these random variables. Our goal is now to construct, as cheaply as possible, Bayes-Nash equilibria (BNE) where agents behave optimally with high probability. In this setting, the cost equation becomes $C_n(u, v) = b \cdot c(u, v) + \min_{P_v} c(P_v) - \E_{A_2}[R_{A_2}(P_{s \to u, v \to t})]$, where the expectation is over $A_2$'s choice of paths.

In this section, we provide a closed form BNE for the $n$-fan. We start with the case of two agents and then consider $m$ competing agents. Since we are searching for BNE, we assume that the agents know their competitor's strategy.

\subsection{Bayes-Nash Equilibria on the \texorpdfstring{$n$}{n}-fan}
As before, let $P_i$ represent the path that includes edge $(v_i, t)$, and let $P_0$ represent the optimal path. Let $\Pr[A_2 \to \mathcal{P}]$ denote the probability that $A_2$ takes any path in the set $\mathcal{P}$. \footnote{For convenience, we write $P_i$ instead of $\{P_i\}$ for singleton sets.} The probability is over $B_2 \sim F$. Let $P_{>i}$ denote the set of paths $\{P_{i+1}, \dots, P_n\}$. Further, note that for any vertex $v_i$, $\argmin_{P_{v_i}} c(P_{v_i}) - \E_{A_2}[R_{A_2}(P_{s\to, v_i\to t})] = (v_i, t)$. This is because going directly to $t$ will always be the cheapest and quickest path from $v_i$ to $t$. Thus:
\begin{align*}
    C_n(v_{i-1}, v_{i}) = c^{i} - \E[R(P_i)] \\
    C_n(v_{i-1}, t) = bc^{i-1} - \E[R(P_{i-1})]
\end{align*}
Further, $\E[R(P_i)] = r/2 \Pr[A_2 \to P_i] + r \Pr[A_2 \to P_{> i}]$. With this in mind, we start with a basic claim:
\begin{proposition} \label{cl:bne}
    When standing at vertex $v_i$, $A_1$ prefers $P_i$ over $P_{i+1}$ if:
    $$r/2\Pr[A_2 \to \{P_i, P_{i+1}\}] > c^i (b_1 - c)$$
\end{proposition}
\begin{proof}
If $A_1$'s expected utility from going immediately to $t$ is higher than their utility from procrastinating, then
    \begin{align*}
        r \cdot \Pr[A_2 \to P_{>i}] &+ r/2 \cdot \Pr[A_2 \to P_i] - b_1c^i > r \cdot \Pr[A_2 \to P_{>i+1}] + r/2 \cdot \Pr[A_2 \to P_{i+1}] - c^{i+1}. \\
        \intertext{It follows that}
        c^i(b_1 - c) & < r(\Pr[A_2 \to P_{>i}] - \Pr[A_2 \to P_{>i+1}]) + r/2 (\Pr[A_2 \to P_i] - \Pr[A_2 \to P_{i+1}]) \\
        & = r(\Pr[A_2 \to P_{i+1}]) + r/2 (\Pr[A_2 \to P_i] - \Pr[A_2 \to P_{i+1}])  \\
        & = r/2\Pr[A_2 \to  \{P_i, P_{i+1}\}]. 
    \end{align*}
\qed
\end{proof}
This claim matches intuition; due to the tying mechanism, $A_1$ gets a reward $r/2$ higher by going to $t$ from $v_i$ if $A_2$ either takes $P_i$ or $P_{i+1}$. Note that $A_1$ is comparing $P_i$ to $P_{i+1}$ because the latter describes what he (naively) believes he will do if he procrastinates. We now describe a Bayes-Nash equilibrium in the $n$-fan.
\begin{theorem}\label{thm:fanBNE}
    Let $G$ be an $n$-fan with reward $r$ and suppose $B_1, B_2$ are drawn from distribution $B$ with CDF $F$. Let $p$ be the solution to $F(\frac{rp}{2}+c) = p$. If $p > \frac{1}{c^{n-1}+1}$, then the following strategy is a Bayes-Nash equilibrium:
    \begin{align*}
        P(b) = \begin{cases}
            \text{take }P_0, & b \le \frac{rp}{2}+c\\
            \text{take }P_n, & \text{otherwise}\\
        \end{cases}
    \end{align*}
    In this equilibrium, for either agent, the probability that they take $P_0$ is $p$. So the expected cost ratio will be $p + (1-p)c^n$.
\end{theorem}
\begin{proof}
    Using \autoref{cl:bne}, at $s$, $A_1$ prefers $P_0$ over $P_1$ if $r/2\Pr[A_2 \to \{P_0, P_{1}\}] > b_1 - c$. Since $\Pr[A_2 \to P_1] = 0$, this simplifies to $b_1 < \frac{r\Pr[A_2 \to P_0]}{2} + c$. Let $\Pr[A_2 \to P_0] = p$. Then, since the biases are drawn iid, for this strategy to be a symmetric BNE we need:
    \begin{align*}
        \Pr[A_1 \to P_0] &= \Pr[A_2 \to P_0] \\
        \Pr[B_1 \le \frac{rp}{2} + c] &= p \\
        F\left(\frac{rp}{2} + c\right) &= p
    \end{align*}
    So if $p$ satisfies this equation, by \autoref{cl:bne}, $A_1$ prefers to go from $s$ to $t$ directly. Now, consider the case where $b_1 > \frac{rp}{2} + c$. Then, $A_1$ prefers to go to $v_1$. Again using \autoref{cl:bne}, at $v_1$, $A_1$ will prefer $t$ over $v_2$ if $r/2\Pr[A_2 \to \{P_1, P_2\}] > c(b_1 - c)$, which implies that $b_1 - c < 0$.
    Since we know that $b_1 > \frac{rp}{2} + c$, it's certainly the case that $b_1 > c$, so this never holds. Thus, $A_1$ prefers to go to $v_2$. The same argument can be repeated until the agent reaches $v_{n-1}$. Then, the agent will prefer $t$ over $v_n$ if $r/2\Pr[A_2 \to \{P_{n-1}, P_n\}] > c^{n-1}(b_1 - c)$, i.e. if  $b_1 < \frac{r(1-p)}{2c^{n-1}} + c$.

    Combining our assumption that $b_1 > \frac{rp}{2} + c$ with this inequality, we get that $p$ must be less than $\frac{1}{c^{n-1}+1}$. When this is false, the agent will always prefer $P_n$ when their bias $b_1 > \frac{rp}{2} + c$. So, under these conditions, the given strategy is a BNE.
\qed
\end{proof}
Notice that while the trivial solution $p = 0$ satisfies $F(\frac{rp}{2} + c) = p$, this is not above $\frac{1}{c^{n-1}+1}$, so the trivial solution is not relevant for finding Bayes-Nash equilibria. And while it's possible (for some distributions $B$ and very low rewards $r$) for the non-trivial solution to $p = \Pr[B_1 \le \frac{rp}{2} + c]$ to be less than $\frac{1}{c^{n-1}+1}$, for the distributions we have explored this is not the case. One might wonder if there's a straightforward generalization of this BNE to other graphs with a dominant path, as in the case without bias uncertainty. In Appendix \ref{app:generalBNE}, we discuss challenges that we encountered trying to do this.

We now use the theorem to understand how much reward is required for optimal behavior with high probability, or a low expected cost ratio (which is a much stronger requirement). Since the expected cost is $p + c^n (1-p)$, in order for this to be low, $1-p$ has to be close to $1/c^n$. Plugging this in to the CDF, we see that for this to happen, we must have
\begin{align*}
    F\left(\frac{r}{2}\left(1-\frac{1}{c^n}\right) + c\right) = 1-\frac{1}{c^n}
\end{align*}
which essentially requires that exponentially little probability mass (in $n$) remains after $r/2$ distance from $c$. For an exponential distribution, this requires $r$ to be linear in $n$, and with a heavier tailed distribution like the Equal Revenue distribution, this requires $r$ to be exponential in $n$. But we may be content with simply guaranteeing optimal behavior with high probability. In that case, so long as $r$ is increasing in $n$, the agents will take the optimal path with high probability for at least the equal revenue, exponential, and uniform distributions. We more precisely explore the probability of optimal behavior and the cost ratios for these distributions in Appendix \ref{app:bneCosts}.

\subsection{Increasing Number of Competitors}
We saw earlier that increasing the number of competitors doesn't change the \textit{per-agent} reward needed for optimal behavior. But one might hope that in Bayesian settings such as bias uncertainty, increasing the number of agents significantly decreases the per-agent reward needed to encourage optimal behavior -- in particular, as the number of agents increases, the probability of \textit{some} agent having a very low bias increases. We provide evidence against this claim. We start by showing that the following strategy is a Bayes-Nash equilibrium with $m$ agents.
\begin{theorem}\label{thm:bneManyAgent}
    Let $G$ be an $n$-fan with reward $r$ and suppose there are $m+1$ competitors with biases drawn iid from distribution $B$ with CDF $F$. Let $p$ be the solution to $F(rd(p,m) + c) = p$, where:
    \begin{align*}
        d(p,m) = \frac{1 - (1-p)^{m}(1+pm)}{p(m+1)}
    \end{align*}
    If $p > m^{-1}\log\left(1+m + \frac{m+1}{2c^{n-1}}\right)$, the following strategy is a Bayes-Nash equilibrium where the probability of optimal behavior for any agent is $p$:
    \begin{align*}
        P(b) = \begin{cases}
            \text{take }P_0, & b \le rd(p,m) + c\\
            \text{take }P_n, & \text{otherwise}\\
        \end{cases}
    \end{align*}
\end{theorem}
\begin{proof}
    See Appendix \ref{app:moreAgents}.
\end{proof}
Now, we consider the equal revenue distribution. In the case of two agents, we required exponentially high reward to get a constant expected cost ratio. Unfortunately, this remains true even as the number of agents competing goes to infinity.
\begin{theorem}\label{thm:erdManyAgent}
    Suppose the biases are drawn from an equal revenue distribution shifted over by $c$. Let $s$ denote the per agent price, so $r = (m+1)s$. Then, no matter how many agents are competing, for any fixed agent, the probability that they behave optimally is at most $\frac{\sqrt{4s+s^2} - s}{2}$. 
\end{theorem}
\begin{proof}
    See Appendix \ref{app:moreAgents}.
\end{proof}
Earlier results show that the probability of optimal behavior with two competing agents is $\frac{s-1}{s}$, which is nearly identical to the theorem's bound as $s$ grows.\footnote{To see this, note that $\frac{\sqrt{4s+s^2} - s}{2}$ can be written as $1-\frac{\sqrt{4s+s^2} - s}{\sqrt{4s+s^2} + s}$ and $\frac{s-1}{s} = 1 - \frac{1}{s}$. Further, $\frac{\sqrt{4s+s^2} - s}{\sqrt{4s+s^2} + s} \approx \frac{2}{2+2s}$ as $s$ grows, since $\sqrt{4s + s^2} \approx s + 2$. So, both probabilities are $1 - O(s^{-1})$.} So even as $m \to \infty$, the relationship between the reward and the probability of optimal behavior doesn't significantly change (and thus the relationship between the reward and expected cost ratio won't significantly change).
We conjecture that, in general, increasing the number of agents does not significantly decrease the per-agent reward. 

\section{Conclusion}
We studied the impact of competition on present bias, showing that in many settings where naive agents can experience exponentially high cost ratio, a small amount of competition drives agents to optimal behavior. This paper is a first step towards painting a more optimistic picture than much of the work surrounding present bias. Our results highlight why, in naturally competitive settings, otherwise biased agents might behave optimally. Further, task/mechanism designers can use our results to directly alleviate the harms of present bias. This competitive model might be a more natural model than other motivation schemes, such as internal edge rewards, and is able to more cheaply ensure optimal behavior. Our work also leaves open many exciting questions.

First, with bias uncertainty, we only obtain concrete results on the $n$-fan. So one obvious direction is to determine which graphs have Bayes-Nash equilibria on the optimal path, and what these equilibria look like. Second, we explore two ``dimensions'' of competition -- the amount of reward and the number of competitors, finding that the latter is unlikely to be significant. Another interesting goal is thus to explore new dimensions of competition.

Lastly, we could extend our work beyond cost ratios, moving to the model where agents can abandon their path at any point. For one, this move would allow us to integrate results on \textit{sunk-cost} bias, represented as an intrinsic cost for abandoning a task that's proportional to the amount of effort expended. \citet{kleinberg2017multiple} show that agents who are sophisticated with regard to their present-bias, but naive with respect to their sunk cost bias can experience exponentially high cost before abandoning their traversal (this is especially interesting because sophisticated agents without sunk cost bias behave nearly optimally). Can competition alleviate this exponential worst case? There are also interesting computational questions in this model. For instance, given a fixed reward budget $r$, is it possible to determine in polynomial time if one can induce an equilibrium where both agents traverse the graph? Such problems are NP-hard for other reward models, but may be tractable with competition. Overall, the abandonment setting has several interesting interactions with competition that we have not explored.

%
%
\bibliographystyle{abbrvnat}
\bibliography{refs}

\appendix
\section{Concrete Distributions for Bias Uncertainty}\label{app:bneCosts}
In this section, we more precisely describe the relationship between $r$ and $p$ (the probability of optimal behavior in the BNE from \autoref{thm:fanBNE}) for several distributions.
\begin{lemma}\label{lem:erdBNE}
    Suppose $B_1$ and $B_2$ are drawn from an equal revenue distribution that's shifted over by $c$ (so $b_1, b_2 > c$). Then, for any $r \ge 2$, the probability of optimal behavior in the BNE from \autoref{thm:fanBNE} is $1 - \frac{2}{r}$, and thus the expected cost is $1 - \frac{2}{r} + \frac{2c^n}{r}$.
\end{lemma}
\begin{proof}
    The equal revenue distribution described has CDF $F(z) = 1 - \frac{1}{z-c+1}$. Now, applying \autoref{thm:fanBNE}:
    \begin{align*}
        1 - \frac{1}{\frac{rp}{2} + c -c+1} &= p\\
        \frac{rp}{2} + 1 - 1 &= p \left(\frac{rp}{2} + 1\right) \\
        \left(\frac{r}{2}\right)p^2 + \left(1-\frac{r}{2}\right)p&= 0 \\
        p &= 0 \text{ or } \frac{r-2}{r}
    \end{align*}
    So with reward $r \ge 2$, $p = \frac{r-2}{r}$.
\qed
\end{proof}
Because this distribution is heavy tailed, the expected cost remains exponential unless $r$ is exponential. However, if $r$ is any increasing function of $n$, then w.h.p. both agents will take the optimal path. Now, we look at the exponential distribution.
\begin{lemma}\label{lem:expBNE}
    Suppose $B_1$ and $B_2$ are drawn from an exponential distribution, $\mathbf{Exp}(\lambda)$, that's shifted over by $c$ (so $b_1, b_2 > c$). Then, for any $r$, the probability of optimal behavior in the BNE from \autoref{thm:fanBNE} is at least $p \ge 1 - O(e^{-\lambda r/2})$, and thus the expected cost is $1 - O(e^{-\lambda r/2}) + O(c^n e^{-\lambda r/2})$.
\end{lemma}
\begin{proof}
    We first solve for $p$ in \autoref{thm:fanBNE}:
    \begin{align*}
        \Pr[B_1 \le \frac{rp}{2} + c] &= p \\
        1-e^{-\lambda(\frac{rp}{2} + c - c)} &= p \\
        1-e^{-\frac{\lambda rp}{2}} &= p \\
        \frac{\ln(1-p)}{p} &= \frac{-\lambda r}{2} \\
        p &= 0 \text{ or } 1 + \frac{2W(-\lambda r /2 \cdot  e^{-\lambda r/2})}{\lambda r}
    \end{align*}
    Where $W$ refers to the Lambert $W$ function, i.e. $W(z) = x \iff xe^x = z$. The Taylor series of $W$ around $0$ is:
    \begin{align*}
        W(x) = \sum_{n=1}^\infty \frac{(-n)^{n-1}}{n!}x^n = x - x^2 + \frac{3}{2}x^3 - \dots
    \end{align*}
    For our application, notice that $x = -\lambda r /2 \cdot  e^{-\lambda r/2}$ is very small; thus, we let $W(x) = x - O(x^2)$. So:
    \begin{align*}
        p \ge 1 - O(e^{-\lambda r/2})
    \end{align*}
\qed
\end{proof}
This means that if $r = O(n/\lambda)$, we get a constant expected cost ratio (and extremely high probability of optimal behavior). We finally consider the simple example of a uniform distribution.
\begin{lemma}\label{lem:uniBNE}
    Suppose $B_1$ and $B_2$ are drawn from the uniform distribution, $\mathbf{Unif}[c, d]$ (so $b_1, b_2 > c$). Then, for $r \ge 2(d-c)$, the probability of optimal behavior in the BNE from \autoref{thm:fanBNE} is $1$. For all $r < 2(d-c)$, this probability is zero.
\end{lemma}
\begin{proof}
    For the first part of the lemma, note that if $rp/2 + c \ge d$, then $F(rp/2 + c) = 1$, so $p = 1$ holds if $r/2 + c \ge d$, or equivalently, $r \ge 2(d-c)$. If we assume $rp/2 + c < d$, then using CDF $F(z) = \frac{z-c}{d-c}$ to solve for $p$:
    \begin{align*}
        \frac{\frac{rp}{2} + c - c}{d - c} &= p \\
        \frac r2 p &= (d-c)p
    \end{align*}
    Since $r < 2(d-c)$, this is satisfied only with $p = 0$. 
\qed
\end{proof}
More generally, for any distribution bounded by $b^*$, computing the reward that guarantees a Nash equilibrium if both agents had public biases $b^*$ also suffices to guarantee a BNE in this setting (and in particular, this BNE always results in optimal behavior). For the case of the uniform distribution, this is the only reward that ever results in optimal behavior. As a direct consequence of this lemma, the expected cost ratio is $1$ with the reward $r = O(d)$.

\section{Difficulties Finding Bayes-Nash Equilibria in Arbitrary Graphs}\label{app:generalBNE}
The simplest way to generalize our BNE from the $n$-fan would be to consider the following schematic:
\begin{align*}
    S_y(b) = \begin{cases}
        \text{take }O, & b \le y\\
        \text{take }P,\text{ for some fixed }P & \text{otherwise}\\
    \end{cases}
\end{align*}
Where $y$ ideally would depend on only the costs in the graph, the reward $r$, and $p$, the solution to $F(y) = p$.
So the agents would either take the optimal path, if their bias was sufficiently small, or some path $P$ that can ``catch'' agents with high bias. In the fan graph, this was the path around the fan. That path happened to be the natural path for any bias, and perhaps more importantly, the ``limiting'' natural path as $b \to \infty$ (i.e. the greedy/myopic path, where you simply take the lowest cost edge at each step). But neither of these ideas generalize; this simple schematic will not produce a BNE in an arbitrary graph, as the following counter example shows.

Consider the modified $3$-fan below. Let $1 < c_1 = c < c^2 < c_2 < c_2^2 < c_3$. So the cost of the edges to $t$ increase faster than doubling.
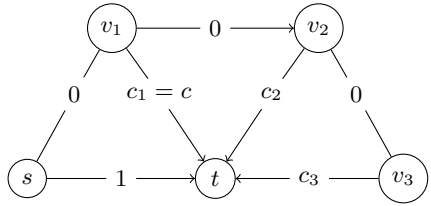
\begin{figure}[H]
    \centering
    \ctikzfig{mod2fan}
    \label{fig:mod2fan}
\end{figure}
We claim the following:
\begin{proposition}\label{cl:bne_counter_ex}
    There is no symmetric BNE for this graph that assigns positive probability to only two paths, assuming the bias is drawn from a distribution with support on $[1, \infty]$.
\end{proposition}
The main idea behind the proof is to notice first that path $P_0$ and $P_3$ will always be taken, when the bias is very low or very high. We then define the interval of biases (the interval depends on the reward, costs, and probability of optimal behavior) wherein $A_1$ wants to take $P_1$ or $P_2$ and argue that at least one of the intervals must be non-empty. 

\begin{proof}    
    Note that as $b:= b_1 \to \infty$, for any fixed reward, the agent will always take the greedy option of choosing the cheapest edge at every step. This will cause them to take the path $P_3$. Further if their bias is below $c_1$, they will take the optimal path, $P_0$. So, any BNE must assign positive probability to those two paths. We now prove, by contradiction, that it must assign positive probability to either $P_1$ or $P_2$.

    Let $p$ be $\Pr[A_2 \to P_0]$ and $1-p = \Pr[A_2 \to P_3]$. At $s$, $A_1$ prefers to procrastinate to $v_1$ when $b > c + rp/2$, so suppose this is the case. At $v_1$, $A_1$ would take $P_1$ if:
    \begin{align*}
        bc - r(1-p) &< c_2 - r(1-p) \\
        b &< c_2/c
    \end{align*}
    So, if $b \in [c+rp/2, c_2/c]$, $A_1$ would take $P_1$. If this interval is non-empty, then we're done. If it's not, then:
    \begin{align*}
        c_2c &> c+rp/2 \\
        rp &> \frac{2c_2}{c} - 2c \\
        r &> \frac{2c_2 - 2c^2}{cp}
    \end{align*}
    If the reward is at least this high, then $A_1$ always goes to $v_2$ at $v_1$. Now, at $v_2$ (where we know $b > c_2/c$), $A_1$ takes $P_2$ when:
    \begin{align*}
        bc_2 - r(1-p) &< c_3 - \frac{r}{2}(1-p) \\
        bc_2 &< c_3 + \frac{r}{2}(1-p) \\
        b &< \frac{2c_3 + r(1-p)}{2c_2}
    \end{align*}
    As before, if $b \in [c_2/c, \frac{2c_3 + r(1-p)}{2c_2}]$, $A_1$ would take $P_2$. We claim that this interval must be non-empty. Suppose it wasn't, and $c_2/c > \frac{2c_3 + r(1-p)}{2c_2}$. Then:
    \begin{align*}
        r(1-p)+2c_3 &< 2c_2^2/c\\
        r &< \frac{2c_2^2 - 2c_3c}{c(1-p)}
    \end{align*}
    However, this upper bound causes a contradiction -- from the definitions, we know that $c_2^2 < c_3$, so clearly $c_2^2 < c_3c$, and thus this bound requires $r < 0$. This contradicts our typical requirement that $r \ge 0$, but even ignoring this requirement, this also contradicts our earlier bound that $r > \frac{2c_2 - 2c^2}{cp}$.
    
    So by contradiction, $A_1$ will either take $P_1$ or $P_2$ with positive probability.
\qed
\end{proof}
In general, the main challenge with finding a BNE is finding a set of paths $\{P_i\}$ and bias intervals $\{I_i\}$ such that, for all intervals $I_i$, $A_1$ must want to traverse $P_i$ if their bias falls in $I_i$. The complication is that $A_1$'s decision to traverse $P_i$ depends on the distribution over $A_2$'s choice of path, which is induced by the intervals $I_i$ and the bias distribution. These interweaving constraints are difficult to manage even for simple graphs, as this preliminary result suggests.

\section{Adding More Competitors with Bias Uncertainty}\label{app:moreAgents}
\begin{proof}[Proof of \autoref{thm:bneManyAgent}]
    Suppose all  agents besides $A_1$ follow $P(b)$. Then, at $s$, $A_1$ views procrastinating to $v_1$ as having reward $-c + r\Pr[A_i \to P_n]^m$, using the iid assumption and the fact that $A_1$ will only win if everyone else has a bias which, according to $P(b)$, causes them to procrastinate. On the other hand, going directly from $s$ to $t$  has utility $-b_1 + \E[R]$, where
    $N$ is a random variable measuring how many agents take $P_0$ and
    $R = \frac{r}{N+1}$ is the random variable corresponding to the reward. Notice that $N \sim \mathbf{Bin}(m, p)$, where $p = \Pr[A_i \to P_0]$. We now use the following claim:
    \begin{claim}
        Let $N \sim \mathbf{Bin}(m, p)$. Then:
        $$\E\left[\frac{1}{N+1}\right] = \frac{1-(1-p)^{m+1}}{p(m+1)} $$
    \end{claim}
    \begin{proof}
        First, using the definition of the expectation and simplifying:
        \begin{align*}
            \E\left[\frac{1}{N+1}\right] &= \sum_{i=0}^m \frac{1}{i+1} \binom{m}{i} p^i (1-p)^{m-i} \\
            &= \sum_{i=0}^m \frac{(1-p)^m}{m+1} \binom{m+1}{i+1} \left(\frac{p}{1-p}\right)^i \\
            &= \frac{(1-p)^{m+1}}{p(m+1)} \sum_{i=0}^m \binom{m+1}{i+1} \left(\frac{p}{1-p}\right)^{i+1}
        \end{align*}
        Now, notice that the binomial theorem says:
        \begin{align*}
            \left(1+\frac{p}{1-p}\right)^{m+1} &= \binom{m+1}{0} \left(\frac{p}{1-p}\right)^{0} + \sum_{j=1}^{m+1} \binom{m+1}{j} \left(\frac{p}{1-p}\right)^{j} \\
            \left(1+\frac{p}{1-p}\right)^{m+1} - 1 &= \sum_{i=0}^m \binom{m+1}{i+1} \left(\frac{p}{1-p}\right)^{i+1}
        \end{align*}
        Plugging this back in:
        \begin{align*}
            \E\left[\frac{1}{N+1}\right] &= \frac{(1-p)^{m+1}}{p(m+1)} \left(\left(1+\frac{p}{1-p}\right)^{m+1}\right) \\ 
            &= \frac{\left((1-p)\left(1+\frac{p}{1-p}\right)\right)^{m+1} - (1-p)^{m+1}}{p(m+1)} \\
            &= \frac{1 - (1-p)^{m+1}}{p(m+1)}
        \end{align*}
    \qed
    \end{proof}
    From the claim, we see that $\E[R] = \frac{r(1 - (1-p)^{m+1})}{p(m+1)}$. Now, $A_1$ prefers to go immediately from $s$ to $t$  when:
    \begin{align*}
        -b_1 + \E[R] &> -c + r(1-p)^m \\
        b_1 &< r\left(\frac{1 - (1-p)^{m+1}}{p(m+1)} - (1-p)^m\right) + c \\
        &= r\left(\frac{1 - (1-p)^{m}(1+pm)}{p(m+1)}\right) + c
    \end{align*}
    For ease of notation, let $d(p, m) = \frac{1 - (1-p)^{m}(1+pm)}{p(m+1)}$. As a sanity check, one can show that $d(p, m)$ is always positive (to show that the expected reward will always be positive).
    \begin{claim}
        $d(p,m) > 0$
    \end{claim}
    \begin{proof}
        We want:
        \begin{align*}
            \frac{1 - (1-p)^{m}(1+pm)}{p(m+1)} &> 0 \\
            1 - (1-p)^{m}(1+pm) &> 0 \\
            1 &> (1-p)^{m}(1+pm) \\
            1 &> e^{-pm}(1+pm) \tag{$1+x \le e^x$}\\
            e^{pm} &> 1+pm
        \end{align*}
        And the last line holds by again using the fact that $1+x \le e^x$.
    \qed
    \end{proof}
    We've shown that when the agent's bias is less than $rd(p,m) + c$, they prefer  $P_0$ to $P_n$ (assuming everyone else takes $P_0$ with probability $p$). Now, if their bias is above $rd(p,m) + c$, we know that they will progress to $v_1$. We claim that they will surely procrastinate all the way to $v_{n-1}$. To see why, since we are analyzing a Bayes-Nash equilibrium, the agent assumes that all other agents follow the equilibrium strategy, which means that all their competitors take either $P_0$ or $P_n$. Thus, the rewards from paths $P_1$ to $P_{n-1}$ are all identical. So by the construction of the $n$-fan, the agent procrastinates until $v_{n-1}$. If $p$ is not extremely low, they will then go to $v_n$, as the following claim demonstrates.
    \begin{claim}
        If $p > \log(1+m + \frac{m+1}{2c^{n-1}})/m$ and the agent has bias $b > rd(p,m)+c$, then the agent takes path $P_n$.
    \end{claim}
    \begin{proof}
        We've just shown that if $b > rd(p,m) + c$, they go to $v_{n-1}$. From there, going to $t$ has utility $-b_1c^{n-1}+r(1-p)^m$ and going to $v_n$ has utility $-c^n + r/2(1-p)^m$. So, the agent prefers $t$ when:
        \begin{align*}
            -c^n + r/2(1-p)^m &< -b_1c^{n-1}+r(1-p)^m \\
            (b_1 - c) c^{n-1} &< r/2 \cdot (1-p)^m \\
            b_1 &< \frac{r(1-p)^m}{2c^{n-1}} + c
        \end{align*}
        So, if the agents bias is in the interval $[rd(p,m) + c, \frac{r(1-p)^m}{2c^{n-1}} + c]$, they prefer to deviate from the equilibrium strategy (in \autoref{thm:bneManyAgent}) by going to $v_{n-1}$. We now compute when this interval is empty:
        \begin{align*}
            r\frac{1 - (1-p)^{m}(1+pm)}{p(m+1)} + c &>  \frac{r(1-p)^m}{2c^{n-1}} + c \\
            \frac{1 - (1-p)^{m}(1+pm)}{p(m+1)} &>\frac{(1-p)^m}{2c^{n-1}} \\
            2c^{n-1}(1 - (1-p)^{m}(1+pm)) &> (1-p)^m p(m+1) \\
            2c^{n-1} &> (1-p)^m p(m+1) + 2c^{n-1}(1-p)^{m}(1+pm) \\
            2c^{n-1} &> (1-p)^m (2c^{n-1} + pm(2c^{n-1}+1) + p) \\
            2c^{n-1} &> e^{-pm} (2c^{n-1} + pm(2c^{n-1}+1) + p) \tag{Since $(1-p)^m \le e^{-pm}$}\\
            e^{pm} &> 1 + mp\frac{2c^{n-1}+1}{2c^{n-1}} + \frac{p}{2c^{n-1}}\\
            e^{pm} &> 1 + m\frac{2c^{n-1}+1}{2c^{n-1}} + \frac{1}{2c^{n-1}} \tag{Since $p \le 1$}\\
            e^{pm} &> 1 + m + \frac{1}{2c^{n-1}(m+1)} \\
            p &> \frac{\log(1+m + \frac{m+1}{2c^{n-1}})}{m}
        \end{align*}
        This is thus sufficient for the agent to take path $P_n$.
    \qed
    \end{proof}
    So, by the claim, if $p > \log(1+m + \frac{m+1}{2c^{n-1}})/m$, the following strategy is a Bayes-Nash equilibrium:
    \begin{align*}
        P(b) = \begin{cases}
            \text{take }P_0, & b \le rd(p,m) + c\\
            \text{take }P_n, & \text{otherwise}\\
        \end{cases}
    \end{align*}
    Which is exactly what the theorem states.
\qed    
\end{proof}
\begin{proof}[Proof of \autoref{thm:erdManyAgent}]
    Recall that, for the equal revenue distribution shifted by $c$, $F(z) = 1-\frac{1}{z-c+1}$. Now, we attempt to solve $F(rd(p,m) + c) = p$ for $p$:
    \begin{align*}
        1-\frac{1}{rd(p,m)+1} &= p\\
        rd(p,m) + 1 - 1 &= p(rd(p,m) + 1) \\
        r d(p,m) &= p(rd(p,m) + 1) \\
        r d(p,m)(1-p) &= p \\
        r(1-p) \cdot \frac{1 - (1-p)^{m}(1+pm)}{p(m+1)} &= p \\
        (1-p)\cdot \frac{1 - (1-p)^{m}(1+pm)}{p^2} &= \frac{m+1}{r}
    \end{align*}
     We only need an upper bound on $p$, not to solve this equation   for general $m$. Implicit differentiation shows that $p$ is increasing in $m$ when the above equation is satisfied. =Given this fact, an upper bound for $p$ is simply the limit as $m \to \infty$, which is:
    \begin{align*}
        \frac{1-p}{p^2} &= \frac{1}{s} \\
        s - sp - p^2 &= 0 \\
        p &= \frac{\sqrt{4s+s^2} - s}{2}
    \end{align*}
\qed
\end{proof}

\end{document}